%% file: main.tex
\documentclass[twocolumn]{article}[11pt]
\usepackage{geometry}
\usepackage{physics}
\usepackage{amsfonts}
\usepackage{amsthm}
\usepackage{geometry}
\usepackage{tikz}
\usepackage{blkarray}
\usepackage{braket}
\usepackage{mathrsfs}
\usepackage{amstext} 
\usepackage{array}   
\newcolumntype{L}{>{$}c<{$}} 
\usepackage{subcaption}
\usepackage{authblk}
\usepackage{booktabs}
\usepackage{caption}
\usepackage{algorithm}
\usepackage[noend]{algpseudocode}
\usepackage{titlesec}
\usepackage{multirow}
\usepackage{abstract}

\newtheorem{theorem}{Theorem}

\newtheorem{definition}[theorem]{Definition}

\titleformat*{\section}{\large\bfseries}
\titleformat*{\subsection}{\bfseries}
\titleformat*{\subsubsection}{\bfseries}
\titleformat*{\paragraph}{\bfseries}
\titleformat*{\subparagraph}{\bfseries}

\usepackage[pdftitle=
{Quantum Error Transmutation}]{hyperref}
\hypersetup{
    colorlinks=true
}

\input{preamble}

\usepackage{xcolor}

\title{\textbf{Quantum Error Transmutation}}

\author[1,2,3]{Daniel Zhang\thanks{daniel.zhang@phasecraft.io}}
\author[1,4]{Toby Cubitt\thanks{toby@phasecraft.io}}

\affil[1]{Phasecraft Ltd.}
\affil[2]{St John's College, University of Oxford}
\affil[3]{Mathematical Institute, University of Oxford}
\affil[4]{Department of Computer Science, University College London}

\begin{document}

\date{ \normalsize \today \vspace{-1em}
\begin{abstract}
\normalsize
We introduce a generalisation of quantum error correction, relaxing the requirement that a code should identify and correct a set of physical errors on the Hilbert space of a quantum computer exactly, instead allowing recovery up to a pre-specified admissible set of errors on the code space. We call these quantum \emph{error transmuting} codes. They are of particular interest for the simulation of noisy quantum systems, and for use in algorithms inherently robust to errors of a particular character. Necessary and sufficient algebraic conditions on the set of physical and admissible errors for error transmutation are derived, generalising the Knill-Laflamme quantum error correction conditions. We demonstrate how some existing codes, including fermionic encodings, have error transmuting properties to interesting classes of admissible errors. Additionally, we report on the existence of some new codes, including low-qubit and translation invariant examples.
\end{abstract}}

\maketitle

\section{Introduction}

The protection of quantum information against noise is a central problem in the realisation of practical quantum computation. Quantum error correcting codes, introduced by Shor \cite{shor95}, combined with the theory of fault tolerance \cite{aharonov1997fault,kitaev1997quantum}, may provide a long-term solution. However, their implementation presents a significant challenge in the near future. For example, for effective error correction, a large number of physical qubits are required to encode each logical qubit of information.

Consequently, as we enter the era of noisy, intermediate-scale (NISQ) devices, there has been significant focus on the development of techniques which seek to mitigate errors instead of correcting them completely. These are algorithmic schemes which post-process data from noisy hardware, and are often grouped together under the term quantum error mitigation. See \cite{cai2022quantum} for a recent review.

In this work we take a different approach, introducing a framework we call quantum error \textit{transmutation} (QET). This arises from a natural generalisation of error correction. Instead of codes for which errors may be identified and corrected perfectly, we seek instead codes for which errors may be identified and corrected \textit{up to} errors on the logical qubits which are deemed admissible. This may enable a reduction in the overhead of encoding.

Subsystem codes \cite{kribs2005operator,nielsen2007algebraic} are included in this formalism; their admissible error set forms the entire group of operators preserving a subsystem of the code space. The codes we consider in this work differ in that the admissible error sets are not of this form. Instead of considering logical information to be encoded in a subsystem of the code space by identifying code space states related by gauge group operators, we regard the whole code space as containing the logical information, and the admissible errors as a set of errors to be tolerated.

There are two immediate applications of QET. First, the digital simulation of quantum systems, one of the most promising applications of near-term quantum computing. For the system of interest to be simulated, there may be a particularly `natural' form of noise, perhaps to which it would be subjected when conducting any real world laboratory experiment anyway. Thus, in a simulation, it may not be necessary to correct all physical errors on the qubits, but instead transmute them to the natural errors on the logical qubits. It was shown in \cite{bausch2020mitigating} that for two fermionic encodings (a fermion to qubit mapping taking the form of a stabiliser code \cite{Gottesman:1997zz}), certain low weight Pauli errors correspond to dephasing noise on the fermions. The latter occurs naturally when fermions are placed in a thermal bath. In this work, we leverage this to produce an error transmuting property.

A second application for these codes is for use in algorithms robust to noise of a particular character, where it may only be necessary to transmute physical errors to logical errors to which the algorithm is resilient, instead of correcting them completely. For example, it was demonstrated in \cite{fontana2021evaluating} that certain implementations of variational quantum eigensolver \cite{peruzzo2014variational} simulations were more resilient to phase damping noise versus depolarising noise, using multiple measures of state quality. Similar resilience characteristics were demonstrated for the construction of a set of target states using a layered circuit ansatz.

\paragraph{Summary} In section \ref{sec:qet} we formally define quantum error transmutation. In section \ref{sec:qetconds}, we outline necessary and sufficient algebraic conditions for a set of physical and admissible logical errors in order for error transmutation, in analogy with those for error correction. In section \ref{sec:existingexamples} we demonstrate error transmuting properties for several existing classes of error correcting codes. This includes topological codes, and more interestingly fermionic encodings, of interest due to the application outlined above. Finally, in section \ref{sec:newexamples} we introduce some novel examples of error transmuting codes, which transmute Pauli errors to dephasing errors on the logical qubits. These include low-qubit examples, a CSS example based on a QR code, and a tileable/translation-invariant code.

\section{Quantum Error Transmutation}\label{sec:qet}

In this section, we define what we mean by quantum error transmutation, as a generalisation of quantum error correction. We start by reviewing the latter, setting up the notation required throughout the rest of this work.

\subsection{Review of Error Correction}

An error correcting code is defined by a subspace (the code space) $\cC$ of the Hilbert space  $\cH$. We assume that a noise process can be described by a quantum operation $\cE$.

\begin{definition}
The code $\cC$ corrects the noise channel $\cE$ iff there exists a CPTP map $\cR$ such that $\forall \rho$ with support in $\cC$ one has:
\be
\cR \circ \cE (\rho) \propto \rho.
\ee
If the noise channel is trace-preserving, we may change `$\propto$' in the above to equality.
\end{definition}

\begin{theorem}
Let $P$ be the projector onto $\cC$, and $\{E_i\}$ be the Kraus operators of the quantum operation $\cE$. Then necessary and sufficient conditions for $\cC$ to correct $\cE$ is that:
\be\label{eq:generalcorrection}
PE_i^\dagger E_j P = \al_{ij} P
\ee
for some Hermitian matrix $\al$.
\end{theorem}

The proof can be found e.g.\ in Nielsen and Chuang \cite{nielsen_chuang_2010}. A code $\cC$ which corrects a noise process $\cE$ with Kraus operators $\{E_i\}$ will also correct any noise process whose Kraus operators are linear combinations of the $E_i$.

An important class of QEC codes are stabiliser codes. Let $\cH = (\bC^2)^{\otimes n}$, and denote the Pauli group on $n$ qubits by $\cG_n$. Then an $[n,k]$ stabiliser code $\cC$ is the vector subspace stabilised by an abelian subgroup $S$ of $\cG_n$ (which does not contain $-I$), with $n-k$ generators $\{g_l\}$. Let $N(S)$ denote the normaliser of $S$ in $\cG_n$. For the Pauli group, $N(S)$ is also the centraliser of $S$.
For stabiliser codes, it is a standard result that the quantum error correction conditions \eqref{eq:generalcorrection} simplify:

\begin{theorem}
For a stabiliser code, and $\{E_i\}$ a set of errors in $\cG_n$, the error correcting conditions in \eqref{eq:generalcorrection} for $\cC$ to correct $\{E_i\}$ reduce to:
\be\label{eq:stabiliserconditions}
E_j^{\dagger}E_k \in (G-N(S)) \cup S \quad \forall \, j,k.
\ee
\end{theorem}

The syndrome of an an element $E$ of $\cG_n$ is $f(E)$, where $f : \cG_n \rightarrow (\bZ_2)^k$ is the homomorphism defined by
\begin{equation}
    E \, g_l = (-)^{f(E)}\, g_l\, E.
\end{equation}
The partition of $\cG_n$ into the $2^{n-k}$ cosets of $N(S)$ is a partition into elements of the same syndrome, as $\text{ker}(f) = N(S)$. The conditions \eqref{eq:stabiliserconditions} state that any pair of errors lie either in different cosets of $N(S)$ (different syndrome), or are related by an element of the stabiliser (and so act identically on $\cC$).

The \textit{weight} of an element of $\cG_n$ is the number of terms in the tensor product which are not the identity. The \textit{distance} of a stabiliser code $\cC$ is the minimum weight of an element of $N(S) \backslash S$.

\subsection{Quantum Error Transmutation}

The goal of quantum error transmutation is to modify a given noise channel $\cE$ on the physical qubits into a noise channel $\cM$ on the logical qubits, where $\cM$ takes a particular form which we now describe.
Standard quantum error correction is then the special case in which $\cM$ is the identity channel.

In this work we restrict to the case where $\cC$ is a stabiliser code, with stabiliser $S$. We note that $N(S)/S$ is isomorphic to the logical Pauli group $\overline{\cG}_k$ on the encoded $k$ qubits. Importantly, this isomorphism is not canonical or unique. Once we fix an arbitrary choice of isomorphism, $N(S)/S \cong \overline{\cG}_k $, there is a remaining freedom to apply any element of $\text{Aut}(\ol{\cG}_k)$, corresponding to a choice of Clifford group element on $k$ qubits. As we shall see, this freedom can be used to our advantage in designing error transmuting codes.\footnote{For example, any abelian subgroup of $\ol{\cG}_k$ can be mapped to one generated by logical Pauli $\ol{Z}$ operators (up to phases).}

\begin{itemize}
    \item We define an \textit{admissible error set} $\overline{M} = \{\overline{m}_{\al}\}$, to be a subset of the logical Pauli group $\ol{\cG}_k$. We define an \textit{admissible logical noise channel} (ALNC) $\cM$ as one stabilising the code space $\cC$, whose Kraus operators proportional to elements of the cosets $\{\cA(\overline{m}_{\al})S\} \subset \cG_n$,  where $\cA$ is a choice of isomorphism:
    \begin{equation}
        \cG_k \xrightarrow[]{\cA} N(S)/S.
    \end{equation}
    We will always assume that the identity is an admissible error.

    \item We define a \textit{physical error set} $E = \{E_i\}$ to be a set of errors on $\cH$, which we restrict to lie in the Pauli group $\cG_n$. We will consider physical noise channels with Kraus operators proportional to elements of $\{E_i\}$ (Pauli channels).
\end{itemize}

\begin{definition}\label{eq:qetconditions}
We say that $\cC$ is a quantum error transmuting code (QETC) for a pair of $\{E_i, \overline{m}_{\al}\}$ of physical errors and logical errors if there exists a single recovery operation $\cR$ such that for all physical noise channels $\cE$ with Kraus operators proportional to elements of $\{E_i\}$, we have:
\be
\cR \circ \cE (\rho) = \cM (\rho)
\ee
where $\cM$ is an ALNC for the logical errors $\{\overline{m}_{\al}\}$ as defined above.

\end{definition}

Thus, an error transmuting code converts a set of physical errors to a set of errors on the logical qubits, whose identification depends on a choice of isomorphism $\cA$.

Our definitions above restrict the set of errors (both physical and logical) to Pauli errors. The main properties of the Pauli group leveraged are the self-inverse property (up to phase), and linear independence as a basis for $\text{Mat}_{2^n\times 2^n}(\cC)$. Thus it is possible to extend many of our definitions to sets of errors where this is also true, although we will not dwell on this in this work.

There are a few more important things to note:
\begin{itemize}
    \item $\cR$ is required to correct $\cE$ to \textit{any} ALNC $\cM$. Some of the $\{\ol{m}_{\al}\}$ may in fact have zero probability in $\cM$.
    \item It is not automatic from the definition that $\cR$ will also transmute physical errors which are linear combinations of $\{E_i\}$ to linear combinations of elements of $\{\overline{m}_{\al}\}$. However, we will see that stabiliser QET codes have this property.
    \item The set of QEC codes for a physical error set $\{E_i\}$ is the same as the set of QET codes with trivial admissible logical error set $\overline{M} = I$. In this sense QET is a generalisation of QEC.
\end{itemize}

\section{Quantum Error Transmutation Conditions}\label{sec:qetconds}

We now describe necessary and sufficient conditions on the admissible and physical error sets for a code $\cC$ to be a QETC. For readability, we divide this into the cases where $M$ forms a subgroup of $\overline{\cG}_k$, and when it does not. Of course, the former is a case of the latter, and we will see the conditions collapse appropriately. We ignore global phases of $\{-1,\pm i\}$, and equations should always be understood to hold up to such a global phase.

For a choice of isomorphism $\cA$,  $\cA(\overline{M})$ corresponds to a set of cosets of $S$ in $N(S)$.  We define the lift $M$ to be the elements of $\cG_n$ contained in these cosets. It is the image of $\ol{M}$ under the map:
\begin{equation}
    \ol{\cG}_k \xrightarrow[]{\cA} N(S)/S \xrightarrow[]{\phi_S^{-1}} N(S) \subset \cG_n,
\end{equation}
where $\phi^{-1}_S$ is the preimage map for the quotient by the stabiliser. By assumption, $M$ always contains $S$.

Let us abuse notation also denote the cosets in $\cA(\overline{M})$ by $\{\overline{m}_{\al}\}$. The elements $m_{\al,k} \in \overline{m}_{\al}$ are related by elements of the stabiliser, and thus act identically on the code space $\cC$. In this notation $M = \bigcup_{\al} \overline{m}_{\al}$.

\subsection{Group Case}\label{sec:qetgroups}

If $\ol{M} \subset \ol{\cG}_k$ forms a subgroup, then by a standard result, the lift $M$ forms a subgroup of $N(S)$ (and thus $\cG_n$). That is: $\ol{M} \cong \cA(\ol{M}) = M/S$. The necessary and sufficient conditions for a QET code in this case are then a rather simple modification of the QEC conditions.

\begin{theorem}
$\cC$ is a quantum error transmuting code for a pair $\{E, \ol{M}\}$ if and only if there exists an isomorphism $\cA$ such that the errors obey:
\begin{equation}\label{eq:qetgroups}
E_j^{\dagger}E_k \in (G-N(S)) \cup M \quad \forall j,k,
\end{equation}
for $M$ defined as $\phi_S^{-1} \circ \cA(\ol{M})$ as above.
\end{theorem}

\begin{proof}
Let us first establish sufficiency, which is easy. Equation \eqref{eq:qetgroups} implies any pair of errors are either in different cosets, or they must be related by an element of $M$. Let us define the operation $\cR$ as follows. First, perform a syndrome measurement (i.e.\ a simultaneous measurement of all the stabiliser generators), projecting $\cE$ onto a channel where only errors in a given coset of $N(S)$ in $\cG_n$ are present. One can then apply any of the errors in that coset (or a probabilistic mixture, see later) to obtain an ALNC.

This proof of necessity is similar to that of normal error correction. Let $\{R_\mu\}$ be the Kraus operators of the recovery operation $\cR$. Let $\cE(\rho) = \sum_i p_i E_i \rho E_i^\dagger$, i.e.\ $p_i$ is the probability of an error $E_i$ (so $\sum p_i = 1$). By assumption there is an ALNC $\cM$ such that, for all $\rho \in \cH$:
\be
\sum_{\mu, i} p_i R_{\mu} E_i P \rho P E_i^\dagger R_{\mu}^\dagger = \sum_{\al} q_{\al} m_{\al} P \rho P m_{\al}^{\dagger}.
\ee
In the above we have, without loss of generality, written the ALNC in terms of a single representative $m_{\al}$ in $M$ of each coset $\overline{m}_{\al}$, as all admissible errors in the same coset $\overline{m}_{\al}$ act the same on $\cC$. We may regard $q_{\al}$ as the total probability of an error in $\overline{m}_{\al}$ occurring.

By a standard theorem, see \textit{e.g.} theorem 8.2 in \cite{nielsen_chuang_2010}, we must have:
\be
\sqrt{p_i} R_{\mu} E_i P = \sum_{\al} \sqrt{q_\al} U_{\mu i , \al}  m_{\al} P
\ee
for some unitary $U_{\mu i , \al}$ (treating $\mu i$ as one index), possibly after padding with zero operators. Then:
\bea
  \sqrt{p_ip_j}&\sum_{\mu} PE_i^\dagger R_\mu^\dagger R_{\mu} E_j P \\
= \sum_{\mu }&\sum_
{\al\beta} U_{\mu i ,\al}^* U_{\mu j,\beta} \sqrt{q_{\al}q_{\beta}} P m_{\al}^\dagger m_{\beta} P.
\eea
Note $\cR$ is trace-preserving, and we may choose $p_i, p_j \neq 0$ since the definition of QET states that $\cR$ is a transmuting operation for all noise channels with Kraus operators proportional to elements of $E$. Thus for for some $ C_{ij,\al \beta}$,  we have:
\be\label{eq:bilinear}
PE_i^{\dagger}E_j P = \sum_{\al \beta} C_{ij,\al \beta} m_{\al}^\dagger m_{\beta} P.
\ee

Suppose there is a pair $E_i$ and $E_j$ such that $E_i^{\dagger} E_j \in N(S)\backslash M$, violating \eqref{eq:qetgroups}. Then \eqref{eq:bilinear} states that:
\be
E_i^{\dagger}E_j P = \sum_{\al} k_{\al} m_{\al} P
\ee
for some coefficients $k_{\al}$, since $m_{\al}^{\dagger} = \pm m_{\al}$ (as they are Pauli errors), and $m_{\al} m_{\beta} \in M$ as $M$ is a group. We regard $E_i^{\dagger}E_j P$ and $m_{\al} P$ as elements of $\cG_k$ and linear maps on $\cC$. The elements of $\cG_k$ are a basis of $\text{Aut}\left((\bC^2)^{\otimes k}\right)$ and are in particular linearly independent. The right-hand side is in $\ol{M}$, but the left-hand side is in $\ol{\cG}_k \backslash \ol{M}$, which is a contradiction. Thus, \eqref{eq:qetgroups} must hold.
\end{proof}

\subsubsection*{Recovery Operation}

It is natural to ask which probability distributions over admissible errors $\{\ol{m}_{\al}\}$ are possible to obtain. This clearly depends on a choice of recovery protocol $\cR$. Let us describe now a more general recovery protocol $\cR$ than the one used in the sufficiency proof above. In section \ref{sec:generalcase}, we show the probability distributions $\{q_{\al}\}$ over admissible errors obtainable via this protocol are actually the most general set of probability distributions possible.

Let us now label the physical error set by the cosets they belong in, $E_{\fn} := \{E_{\mathfrak{n},i}\}$, where $i$ now labels the errors in a given coset $\mathfrak{n}$ of $N(S)$. The most general noise channel we consider is described by a set of probabilities $p_{\mathfrak{n},i}$ such that
\be \label{eq:errorchannelgroup}
\cE(\rho) = \sum_{\mathfrak{n},i} p_{\mathfrak{n},i} E_{\mathfrak{n},i} \rho E_{\mathfrak{n},i}^{\dagger}.
\ee
The steps in the recovery protocol are:
\begin{itemize}
    \item A syndrome measurement, projecting onto a channel where only errors in a given coset of $N(S)$ in $\cG_n$ occur. Label the $2^{n-k}$ cosets and their projectors by $\mathfrak{n}$ and $P_{\mathfrak{n}}$ respectively.

    \item Conditional on the measurement, apply a recovery operation $\cR_{\mathfrak{n}}$. The Kraus operators $R_{\mathfrak{n},\al}$ of $\cR_{\mathfrak{n}}$ can be described as follows. Without loss of generality, pick a reference error $E_{\mathfrak{n},1}$ with syndrome $\mathfrak{n}$. For each $\ol{m}_{\al}$, take $R_{\mathfrak{n},\al} = \sqrt{r_{\mathfrak{n},\al}}m_{\al}E_{\mathfrak{n},1}^{\dagger}$, where $\sum_{\al} r_{\mathfrak{n},\al}=1$.  Then we have:
    \be
    R_{\mathfrak{n},\al} E_{\mathfrak{n},j} P 
    = \sqrt{r_{\mathfrak{n},\al}} m_{\al} E_{\mathfrak{n},1}^{\dagger}E_{\mathfrak{n},j}  P
    \ee
    for all $E_{\mathfrak{n},j}$ with the same syndrome.
    Note  $m_{\al} E_{\mathfrak{n},1}^{\dagger}E_{\mathfrak{n},j} \in M$ by \eqref{eq:qetgroups} and the fact that $M$ is a group. Thus $\cR_{\mathfrak{n}}$ maps an error channel with operators in a given coset of $N(S)$ to an ALNC, and thus by linearity it maps all noise channels with operators in $E$ to an admissible channel. The choice of parameters $\{r_{\mathfrak{n},\al}\}$ determine the form of the final ALNC.
\end{itemize}
Thus the total recovery operation is given by:
\be
\cR(\rho) = \sum_{\mathfrak{n}, \al} r_{\mathfrak{n}, \al} m_{\al} E_{\mathfrak{n},1}^{\dagger}P_{\mathfrak{n}} \rho P_{\mathfrak{n}}E_{\mathfrak{n},1} m_{\al}^{\dagger},
\ee
and transmutes \eqref{eq:errorchannelgroup} to an ALNC as above.

\subsection{General Case}\label{sec:generalcase}

Let us now generalise to the case where $\ol{M}$ does not form a group. We continue to label the cosets of $N(S)$ by $\fn$, and errors in a given coset by $E_{\fn} := \{E_{\fn,i}\}$.

\begin{theorem}\label{thm:qetgeneral}
$\cC$ is a quantum error transmuting code for a pair $\{E, \ol{M}\}$ if and only if there exists an isomoprhism $\cA: \ol{\cG}_k \rightarrow N(S)/S$, such that for each $\mathfrak{n}$ there exists a map:
\be
\pi_{\fn}: E_{\fn}  \rightarrow \cA(\ol{M}),\qquad \pi_{\fn}(E_{\fn,i}) := \ol{m}_{\pi_{\fn},i}
\ee
so that $\forall\, E_{\fn,i},\,E_{\fn,j} \in E_{\fn}$:
\be \label{eq:qetcgeneral1}
E_{\fn,i}^\dagger E_{\fn,j}|_{\cC} = {\ol{m}}_{\pi_{\fn},i}^\dagger \ol{m}_{\pi_{\fn},j}.
\ee
In the above, both sides are regarded as linear maps on $\cC$. We have again used the same notation $\ol{m}_{\pi_{\fn},i}$ for elements of $\ol{\cG}_k$ and $N(S)/S$, with the isomorphism $\cA$ understood. Equivalently, writing $m_{\fn,i}$ as a choice of representative of $\ol{m}_{\fn,i}$ in $N(S)$ and $\cG_n$:
\be \label{eq:qetcgeneral2}
E_{\fn,i}^\dagger E_{\fn,j} P =  m_{\pi_{\fn},i}^\dagger m_{\pi_{\fn},j} P.
\ee
\end{theorem}

Some observations:
\begin{itemize}
    \item Plainly, the conditions state that the errors in a given coset embed in a nice way into the admissible errors such that the bilinears are preserved. Note it suffices to check \eqref{eq:qetcgeneral1} or \eqref{eq:qetcgeneral2} for a fixed choice of $i$, whilst varying $j$.
    \item The error-transmuting condition for groups is also a \textit{sufficient} condition even when $\ol{M}$ is not necessarily a group. That is:
    \begin{equation}\label{eq:qetstonggeneral}
    E_j^{\dagger}E_k \in (G-N(S)) \cup M \quad \forall j,k,
    \end{equation}
    implies the conditions \eqref{eq:qetcgeneral1}, even when $\ol{M}$ is not a group.

    To see this, suppose \eqref{eq:qetstonggeneral} holds. For each $\mathfrak{n}$, we can define a $\pi_{\mathfrak{n}}$ by $\pi_{\mathfrak{n}}(E_{\mathfrak{n},1}) = I$. Then \eqref{eq:qetstonggeneral} determines $\ol{m}_{\pi_{\fn},j} = \pi_{\fn}(E_{\mathfrak{n},j}) = E_{\fn,1}^\dagger E_{\fn,j} $.

    We will refer to the sufficient conditions \eqref{eq:qetstonggeneral} as \textit{strong} error transmuting conditions, for when $\ol{M}$ is not a group.

    \item If $\ol{M}$ is a group, the QET conditions coincide. That \eqref{eq:qetcgeneral1} implies \eqref{eq:qetgroups} is obvious. The other direction is given by the argument above.

\end{itemize}

\begin{proof}
($\Leftarrow$) The proof of sufficiency is straightforward. It illustrates the recovery protocol for QET, which we show in the necessity proof to be general enough to reproduce all possible probability distributions over admissible errors. To transmute the error channel to an ALNC, one first performs a syndrome measurement, applying projectors $P_{\fn}$ onto error channels whose Kraus operators belong to the same coset of $N(S)$, i.e.\ the set $E_{\mathfrak{n}}$. For each map $\pi_{\fn}$ satisfying the conditions in the theorem, we may define a recovery operator $R_{\pi_{\fn}}$ such that:
\be
R_{\pi_{\fn}} E_{\fn,i} = m_{\pi_{\fn},i}
\ee
with $m_{\pi_{\fn},i}$ defined as above. For example, one may take without loss of generality $R_{\pi_{\fn}} =m_{\pi_{\fn},1}  E_{\fn,1}^{\dagger}$. This is consistent because:
\bea
R_{\pi_{\fn}} E_{\fn,j} &= R_{\pi_{\fn}} E_{\fn,1} E_{\fn,1}^{\dagger}E_{\fn,j} \\
&= m_{\pi_{\fn},j} g
\eea
for some $g \in S$, by \eqref{eq:qetcgeneral2}. Of course $m_{\pi_{\fn},j} g$ acts as $\ol{m}_{\pi_{\fn},j}$ on the code space $S$. If there are multiple $\pi_{\fn}$ possible, then let $R_{\pi_{\fn}} = \sqrt{r_{\pi_{\fn}}} m_{\pi_{\fn},1}  E_{\pi_{\fn},1}^{\dagger}$ be the Kraus operators for the recovery operation for each coset, where $\sum_{\pi_{\fn}} r_{\pi_{\fn}} = 1$. The total recovery operation $\cR$ is then given by:
\be
\cR(\rho) = \sum_{\mathfrak{n},\pi_{\fn}} r_{\pi_{\fn}} m_{\pi_{\fn},1} E_{\mathfrak{n},1}^{\dagger}P_{\mathfrak{n}} \rho P_{\mathfrak{n}} E_{\mathfrak{n},1} m_{\pi_{\fn},1}^{\dagger}.
\ee
Taking a general noise channel
\begin{equation}\label{eq:noisechannelgeneral}
    \cE(\rho) = \sum_{\mathfrak{n},i} p_{\mathfrak{n},i} E_{\mathfrak{n},i} \rho E_{\mathfrak{n},i}^{\dagger},
\end{equation}
then for all $\rho \in \cC$ the combined error and recovery channel is an ALNC:
\be
\cR\circ\cE(\rho) = \sum_{\fn,i} \sum_{\pi_\fn} p_{\fn,i}r_{\pi_\fn} m_{\pi_{\fn},i} \rho \, m_{\pi_{\fn},i}^{\dagger}.
\ee

($\Rightarrow$) Now we show necessity. First note that without loss of generality, the recovery channel can be expressed as a syndrome measurement and then an application of a syndrome-dependent channel. To see this, let $\{R_{\beta}\}$ be the Kraus operators for $\cR$. Then for $\rho \in \cC$:
\begin{align}
\cR \circ \cE(\rho) &= \sum_{\beta}R_{\beta} \left(\sum_{\fm} P_{\fm} \right)  \cE(\rho) \left(\sum_{\fl} P_{\fl} \right) R_{\beta}^{\dagger} \nonumber \\
&= \sum_{\fm}\sum_{\beta} R_{\beta}^{\fm}P_{\fm} \cE(\rho) P_{\fm} {R_{\beta}^{\fm}}^{\dagger}
\end{align}
where we have denoted $R^{\fn}_{\beta} := R_{\beta}P_{\fn}$, and $\cE$ is of the form \eqref{eq:noisechannelgeneral}. We have used $P_{\fm}E_{\fn,i} \rho E_{\fn,i}^{\dagger} P_{\fl} = \delta_{\fn,\fm} \delta_{\fn,\fl}E_{\fn,i} \rho E_{\fn,i}^{\dagger}$, and reordered the summation. The channel applied, conditional on a syndrome measurement of $\fm$, is thus $\cR_{\fn}(\cdot) = \sum_{\beta} R^{\fm}_{\beta} \cdot  {R^{\fm}_{\beta}}^{\dagger}$.

Now the definition of QET requires that $\cR$ is a valid recovery procedure for any probability distribution $\{p_{\fn,i}\}$ of physical errors. Thus by linearity it must correct the error channel consisting of a single error $E_{\fn,i}$. Let us drop the coset label $\fn$ now for readability. We must have, for all $\rho \in \cH$,
\be
\sum_{\beta}R_{\beta} E_i P \rho P E_i^\dagger R_{\beta}^{\dagger} = \sum_{\al} q^i_{\al}m_{\al}P\rho P m_{\al}^\dagger
\ee
for some $q^i_{\al}$. This implies that:
\be
\sum_{\beta} U_{\al\beta } R_{\beta}  E_i P =  \sqrt{q^i_{\al}} m_{\al} P
\ee
where $U$ is unitary, padding with zero operators if necessary. Thus, one can redefine the Kraus operators $\sum_{\beta} U_{\al\beta } R_{\beta} \rightarrow R_{\al}$, so $R_{\al}E_iP = \sqrt{q^i_{\al}} m_{\al}P$.

Now consider another physical error $E_{j}$ in the same coset as $E_i$ (if it exists). One has:
\bea\label{eq:qetcgeneralproof}
R_\al E_j P &= R_{\al} E_i P E_i^{\dagger}E_jP \\
&= \sqrt{q^i_{\al}} m_{\al} E_i^\dagger E_j P,
\eea
noting that $E_i^\dagger E_j \in N(S)$ so commutes with $P$. The same recovery procedure must also transmute the single $E_j$ error channel to an ALNC:
\begin{align}
\sum_{\al}& R_{\al} E_j P \rho P E_{j}^\dagger R_{\al}^\dagger\\
&= \sum_{\al} q^i_{\al}  m_{\al} E_i^\dagger E_j  P \rho P E_j^\dagger E_i m_{\al}^\dagger  \nonumber\\
&= \sum_{\beta} q^{j}_{\beta} m_{\beta}P \rho P m_{\beta}^\dagger
\end{align}
for some $\{q^{j}_{\al}\}$, so there is a unitary $V$ such that:
\be
\sqrt{q^i_{\al}}  m_{\al} E_i^\dagger E_j P  - \sum_{\beta}V_{\al\beta}\sqrt{q^j_\beta} m_{\beta}P = 0.
\ee
The operators appearing in this sum may be regarded as elements of $\ol{\cG}_k$, which are a basis of the vector space of linear operators on $\cC$. Thus:
\be
 m_{\al} E_i^\dagger E_j P = m_{\beta}P \,\, \Rightarrow \,\, E_i^\dagger E_j P = m_{\al}^{\dagger} m_{\beta}P
\ee
for some $\beta$ (up to phase). This holds for all $E_j$ in the same coset of $N(S)$ as $E_i$. One therefore obtains a map $\pi$ as above (omitting the coset label), where $E_i \mapsto m_{\al}$, and $E_j \mapsto m_{\beta}$ for $\beta$ determined as above. 
\end{proof}

Note in the necessity proof the most general form of the recovery protocol (up to unitary equivalence of Kraus operators) was derived. We see it is identical to the protocol derived in the sufficiency part.

Further, as promised, we see that the recovery procedure in section \ref{sec:qetgroups} for when $\ol{M}$ is a group is the most general recovery procedure. This is because, for a group, the set of maps $\pi_{\fn}$ as in theorem \ref{thm:qetgeneral} is in bijection with the set of elements of $\ol{M}$. For a reference $E_{\fn,1}$ in a given coset, for each element $\ol{m}_{\al}$ one can construct a $\pi_{\fn}$ such that $\pi_{\fn}(E_{\fn,1}) = \ol{m}_{\al}$. The map on the remainder of $E_{\fn}$ is then determined by \eqref{eq:qetcgeneral1}. The most general recovery protocol outlined in the proof above then agrees with that outlined in section \ref{sec:qetgroups}.

\subsection{Effective Code Distance} \label{sec:effcodedistance}

In normal error correction, often the goal is to design codes with high code distance $d$, so that the code can perfectly correct errors up to weight $\lfloor \frac{d-1}{2}\rfloor$. This is in accordance with the simple but oft-used model of noise on quantum computers as being independent and identically distributed qubit noise, so that lower weight Pauli errors are more probable.

For an error-transmuting code $\cC$ for admissible errors $\ol{M}$, we define the \textit{effective} code distance $d_{\text{eff}}=2w+1$, where $w$ is the maximum weight such that the set of all errors with weight $\leq w$ obey the error transmuting conditions \eqref{eq:qetcgeneral1}. This requires the existence of an appropriate isomorphism $\cA$. The code $\cC$ is therefore a QET code for the physical error set $E$ of Pauli errors of weight $ \leq w$.

From the \textit{strong} form of the error transmuting conditions \eqref{eq:qetstonggeneral} (or if $\ol{M}$ is a group), it is easy to identify a lower bound:
\begin{equation}
    d_{\text{eff}} \geq \min_{E \in N(S)\backslash M} \text{weight}(E).
\end{equation}
Thus we have that $d_{\text{eff}} \geq d$. In words, if the code $\cC$ has error correcting distance $d$, then as an error transmuting code one can effectively increase the distance to $d_{\text{eff}}$ if $M\subset N(S)$ contains all the errors in $N(S)\backslash S$ with weight in the interval $[d,d_{\text{eff}})$.

Recall that $M = \phi_S^{-1}\circ\cA(\ol{M})$ depends on a choice of isomorphism $\cA : \ol{\cG}_k \rightarrow N(S)/S$. By the above discussion, in producing error transmuting codes, it is often advantageous to make a tactical choice of $\cA$ such that $M$ `cleans up' all the errors of weight $[d,d_{\text{eff}})$, resulting in a QET code of effective distance at least $d_{\text{eff}}$. We will see this in practice momentarily.

\section{Existing Examples}\label{sec:existingexamples}

In this section we describe how some existing error correcting codes in the literature may also be taken to be error transmuting codes.

We first note that some families of codes are easily seen to be examples of QET codes. For example, subsystem codes \cite{kribs2005operator,nielsen2007algebraic}, for which:
\be
\cH = \cA \otimes \cB \oplus \cC^{\perp}.
\ee
That is, the code space is decomposed into subsystems $\cC = \cA \otimes \cB $, and information is stored in $\cA$ whilst $\cB$ is acted on by a gauge group $\cG_{\text{gauge}}$ which is treated as a redundancy. In our formalism, subsystem codes are simply QET codes with an admissible error group which decomposes as a tensor product: $M = I_{\cA} \otimes \cG_{\text{gauge}}$ consisting of \textit{all} operators in $N(S)$ which act as identity on the subsystem $\cA$. The study of subsystem codes is already deep and varied, so we do not focus on them here, instead considering admissible error sets which do not form such a gauge group.

Trivially, any CSS code is also an error transmuting code if one decides on allowing the group of logical phase errors or the group of logical bit flip errors, see the discussion in section \ref{sec:CSScodes}. This includes of course codes which are essentially just classical codes, such as the repetition code.

\subsection{Fermionic Encodings}

As mentioned in the introduction, one of the most fruitful areas we predict error transmutation to be applicable to is the digital simulation of quantum systems which are subject to `natural' noise, and in which a certain degree of that noise is tolerated. In these situations, it may be possible to encode the system in a way that physical errors are transmuted to natural (admissible logical) errors. This may reduce the overhead of the encoding and could even be desirable if the goal is to simulate a noisy system.

Digital simulation of fermionic many-body systems requires a mapping of the fermionic Fock space, and algebra, to the qubit Hilbert space and Pauli algebra. The result of the mapping is often a stabiliser code, where the fermionic Fock space is realised as the code space. The fundamental trade-off in fermionic encodings is that the non-local statistics manifest either in non-local qubit operators or in long-range entangled states.

In the latter approach, the fermionic operators are low-weight, making their application less costly. However, as they are logical operators, this results in a low code distance. In \cite{bausch2020mitigating} it was demonstrated that despite this, some of these codes retain valuable error mitigating properties. In particular, for the Verstraete-Cirac encoding \cite{verstraete2005mapping} and the compact encoding of Derby \& Klassen \cite{derby2021compact}, despite having undetectable weight one logical operators (and therefore a code distance of~1), these undetectable errors corresponding to dephasing operators on the fermionic system. It was shown that phase noise occurs naturally in a lattice of fermions via the coupling to phonons arising from placement in a thermal bosonic bath.

In this work, we show that one can extend this error detecting property to an error transmuting property. That is, if one accepts single dephasing errors on the simulated fermions as an admissible logical error set, then one can transmute all single qubit Pauli errors on the physical qubits. The distance of the code is enlarged to the effective distance $d_{\text{eff}}=3$.

For the sake of brevity, we will demonstrate the above only for the compact encoding. Consider a square lattice with fermionic sites at the vertices,\footnote{We disregard boundary effects in this paper, leaving the details of this to future work.} labelled by $j$. The corresponding qubit system is described as follows. Take the same lattice but with faces labelled in a checker-board even-odd manner. A vertex qubit is associated to each fermionic site, and a face qubit to the odd faces. The edges of the lattice are given an orientation circulating clockwise and anticlockwise around faces on alternating rows. The unique odd face adjacent to an edge $(i,j)$ is denoted $f(i,j)$. The arrangement is displayed in figure \ref{fig:compact_encoding}.

The even fermionic operator algebra\footnote{All physical fermionic observables are even fermionic operators by parity superselection.} is generated by edge and vertex operators:
\begin{equation}
    E_{jk} := -i\gamma_j\gamma_k,\quad V_j := -i\gamma_j\ol{\gamma}_j,
\end{equation}
where $\gamma_j, \ol{\gamma}_j$ are Majorana operators associated to the fermion at site $j$. The encoding maps these operators to:
\begin{equation*}
        \tilde{E}_{jk} := \begin{cases}
        \,\,\,\,\,X_jY_k X_{f(j,k)}\quad &(j,k) \text{ downwards},\\
        -X_jY_k X_{f(j,k)}\quad &(j,k) \text{ upwards},\\
        \,\,\,\,\,X_jY_k Y_{f(j,k)}\quad &(j,k) \text{ horizontal}\\
    \end{cases}
\end{equation*}
\begin{equation}
    \\
    \tilde{E}_{kj} := -E_{jk} \qquad \tilde{V}_j := Z_j.
\end{equation}
The mapped operators obey the same local (anti)commutation relations as the fermionic operators. There is also a non-local relation in that any product of a loop of edge operators must equal the identity (up to a phase). In the encoding, this is enforced by projecting to the simultaneous $+1$ eigenspace of all loops of encoded edge operators, which are the stabilisers of the code. It is sufficient to do so for the minimal generating set of loops given by those around each face. The odd faces give trivial stabilisers, whilst those associated to even faces are non-trivial and are given in figure \ref{fig:compact_encoding}.

\begin{figure}[ht!]
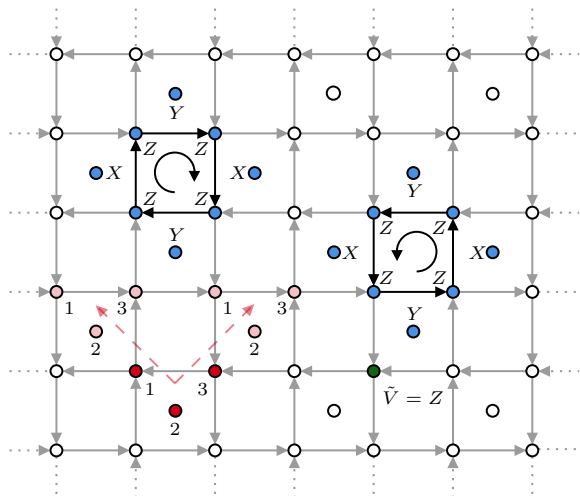

    \centering

\tikzset{every picture/.style={line width=0.75pt}} 


    \caption{The compact encoding. Two stabilisers are pictured (in blue), as well as the unit cell and its translations (in red). The admissible errors are $Z$ errors on the vertex qubits, one of which is displayed (in green).}
    \label{fig:compact_encoding}
\end{figure}

The fermionic dephasing operator is the vertex operator $V_j$. Thus the error transmuting property we claim is precisely that all single qubit Pauli errors may be transmuted to errors in the same stabiliser coset as $\tilde{V}_j$.

This is easy to check. We note that all qubits can be classified as one of three distinct types (indicated in red in figure \ref{fig:compact_encoding}). The grouping of three may be considered a unit cell, which is then tiled along axes at $\pm 45^{\circ}$ to the lattice. From this rotated perspective, the qubit arrangement is known as the Lieb (or decorated square) lattice.\footnote{One can show that the compact encoding is equivalent up to single qubit Cliffords to a translation invariant code in the Lieb lattice \cite{chien2022optimizing}. One can then check the error transmuting properties described in the main text explicitly using the algebraic description of such codes, see section \ref{sec:tileable}.}

Within a given unit cell, the \textit{non-zero} syndrome cosets are: $\{X_1, Y_1\}$, $\{X_2\}$, $\{Y_2\}$, $\{Z_2\}$, $\{X_3,Y_3\}$. Both $Z_1$ and $Z_3$ have zero syndromes. Further, the syndromes for errors with non-zero syndrome in any two different unit cells are distinct; they induce different stabiliser measurements. Assuming a single qubit Pauli error occurs, there are three possibilities.
\begin{itemize}
    \item The syndrome measurement returns zero, indicating either no error, or an error of type $Z_1$ or $Z_3$. These errors are already of the admissible type $\tilde{V}_j$.
    \item The unique syndrome of any Pauli error on a qubit of type $2$ is measured. These errors may be corrected exactly.
    \item A syndrome indicating an $X$ or $Y$ error on a qubit of type $1$ or $3$ has occurred. A channel probabilistically applying $X$ and $Y$ (or deterministically either) will transmute the errors to $Z$ on the same qubit, which are admissible.
\end{itemize}
We conclude the existence of the error transmuting property as desired.

\subsection{Topological Codes}

It is intuitively clear that having a precise description of the form of the logical operators of a code aids significantly in the design of error transmuting codes. A large class of examples for which this is true are topological codes.

In the interest of brevity, we restrict our discussion to topological codes associated with a 2d Riemann surface $\cM$.\footnote{We will be brief in our exposition of these codes. See e.g.\ \cite{bombin2013introduction} for detailed explanations.} These are specified by the $\bZ_2$ homology of the surface, and depend on a specific cellulation. The $0,1,2$-cells are the vertices, edges and faces respectively. A qubit is assigned to each $1$-cell (equivalently the $1$-cocell intersecting it). Tensor products $Z(c)$ of Pauli $Z$ operators are associated with $1$-chains $\{c\}$, and $X$ operators with $1$-cochains $\{\tilde{c}\}$.

The stabiliser generators consist of:
\begin{itemize}
    \item Plaquette generators $Z(\partial c)$ defined by the boundary of a $2$-cell $c$.
    \item Vertex generators $X(\tilde{\partial}\tilde{c})$ defined by coboundaries of $0$-cocells $\tilde{c}$.
\end{itemize}
Here $\partial$ and $\tilde{\partial}$ are the boundary and coboundary operators in the (co)chain complexes:
\begin{align}
     &C_2 \xrightarrow[]{\partial_2} C_1 \xrightarrow[]{\partial_1} C_0  \\
    &C^0 \xrightarrow[]{\tilde\partial^0} C^1 \xrightarrow[]{\tilde\partial^1} C^2
\end{align}
where $C_i$ and $C^i$ are the $i$-chains and cochains respectively. They require a careful definition in the case when $\cM$ is a manifold with boundary.

The logical operators $N(S)$ are generated by the elements of $Z_1 := \text{ker} \,\partial_1$ and $Z^1 := \text{ker} \,\tilde\partial^1$, and the equivalence classes $N(S)/S$ by the (co)homology groups $H_1 : = Z_1/\text{im}\, \partial_2$ and $H^1 : = Z^1/\text{im}\, \tilde\partial^0$. The number of encoded logical qubits is equal to the first Betti number $\beta_1 = \text{dim}\, H_1$.

For elements $[c]$ of $H_1$ denote:
\begin{equation}
    d([c]) = \min_{c' \in Z_1 | [c']=[c]} \text{weight}(c'),
\end{equation}
and similarly for $H^1$. There is a labelling of elements of $H_1$  $[c_i]$, $i=1,\ldots,2^{\beta_1}$, where if $d_i:=d([c_i])$ is the minimum weight of an operator in the same homology class of $c_i$ (that is, in the same stabiliser coset in $N(S)/S$), then $0=d_1 \leq d_2 \leq\ldots\leq d_N$. There is similarly a labelling of elements $\tilde{c}^i$ of $H^1$, and a sequence $0=d^1 \leq d^2 \leq\ldots\leq d^N$. Note $c_1$ and $\tilde{c}^1$ correspond to the identity operator, hence $d_1=d^1=0$.

The distance of the code as a regular error correcting code is simply $\text{min}(d_2, d^2)$. However, if there exists a choice of identification $N(S)/S \cong \ol{G}_k$ with the logical Pauli group such that the admissible logical errors are precisely those errors $c_1, c_2,\ldots, c_j$ and $\tilde{c}_1, \tilde{c}_2,\ldots, \tilde{c}_{j'}$ for some $j, j'$, then the code is able to transmute all errors of weight $\lfloor \frac{\min(d_j,d_{j'})-1}{2}\rfloor $ to that admissible error set.

Although much of the discussion above holds more generally, for topological codes the ordering $c_1,c_2,\ldots$ often interacts nicely with the weights of these operators on the \textit{logical} qubits. For example, for the toric code \cite{Kitaev:1997wr} based on an $L\times L$ periodic square lattice, $\ol{Z}_1$ and $\ol{Z}_2$ are identified with the two loops giving the canonical basis of $H_1 = \bZ_2 \oplus \bZ_2$ and one has that $d(\ol{Z}_1)=d(\ol{Z}_2)=L$, $d(\ol{Z}_1\ol{Z}_2)=2L$. Similar statements hold for the $\ol{X}$ operators. One can easily see how this argument generalises to higher genus $g$ Riemann surfaces, or surfaces with boundary.

Taking lattices which are not self-dual results in asymmetric CSS codes (see section \ref{sec:CSScodes}). Supplemented with the geometric origin of logical operators, this can result in nice error transmuting codes. For example, take the toric code with a cellulation given by an $L \times L$ hexagonal lattice, where $L$ is even. The dual cellulation is the triangular lattice. If $\ol{Z}_1$ and $\ol{Z}_2$ are identified with the generators of $H_1 = \bZ_2 \oplus \bZ_2$ and $\ol{X}_1$ and $\ol{X}_2$ with those of $H^1 = \bZ_2 \oplus \bZ_2$, we have:
\bea
    &d(\ol{Z}_1)=d(\ol{Z}_2)=2L,\quad d(\ol{Z}_1\ol{Z}_2) = 3L,\\
    &d(\ol{X}_1)=d(\ol{X}_2)=L,\quad d(\ol{X}_1\ol{X}_2) = \frac{3}{2}L.
\eea
If we accepted single logical qubit bit-flip errors, then the effective distance of the code is increased from $L$ to $d_{\text{eff}}=\frac{3}{2}L$. That is, the code can transmute Pauli errors of weight $\lfloor \frac{3}{4}L-\frac{1}{2} \rfloor$ to $\{\ol{X}_1,\ol{X}_2\}$.

\section{New Examples}\label{sec:newexamples}

In this section, we report the existence of some novel quantum error transmuting codes which do not fall into the above classes. For simplicity, we will primarily take the physical error set to consist of Pauli errors up to a certain weight, and the admissible logical set to be single (logical) qubit Pauli phase errors. Apart from the low qubit examples, the codes below are not necessarily optimal in their parameters, but hopefully serve as a proof of concept.

\subsection{Low Qubit Examples}\label{sec:lowqubitexample}

\subsubsection*{Group Case}

The stabiliser of a code encoding two qubits in seven is given in table \ref{tab:7qubitcode}. It transmutes all single qubit Pauli errors to the group generated by a logical phase error $\ol{Z}_1$ on one of the logical qubits.

\begin{table}[ht!]
\begin{center}
\begin{tabular}{ L | L L L L L L L}
 & X & X & Y & Y & Z & I & Z \\
 &I & Z & X & Y & Y & X & Y \\
S &I & I & I & I & I & Z & Z \\
 &Z & Z & I & I & Z & I & Z \\
 &Z & Z & Z & Z & I & I & I \\
 \hline\rule{0pt}{1\normalbaselineskip}
\ol{X}_1 &I & X & X & I & X & I & I \\
\ol{X}_2 &I & I & X & X & I & I & Z \\
 \hline\rule{0pt}{1\normalbaselineskip}
 \ol{Z}_1 &Z & Z & I & I & I & I & I \\
\ol{Z}_2  &Z & I & I & Z & I & I & Z \\
\end{tabular}
\caption{Stabilisers and logical Pauli operators for the seven qubit code.}\label{tab:7qubitcode}
\end{center}
\end{table}

This code was found by searching over the standard form of stabiliser codes \cite{Gottesman:1997zz, Cleve:1996pp}, which expresses the stabiliser generators as a $(n-k) \times 2n$ matrix over $\mathbb{F}_2$.  The standard form also outputs a choice of generators of $N(S)/S$ which are identified with those of the logical Pauli group $\bar{\cG}_k$ up to symplectic automorphism (i.e.\ conjugation by Clifford unitary). If one wanted to search for codes which transmuted single qubit Paulis to a single logical phase error on one of the qubits, one could therefore search for codes such that all two qubit errors in $N(S)$ are contained in a single coset within $N(S)/S$.

For the above code, all single qubit Pauli operators have non-zero syndrome,\footnote{Thus, the code may also be used as a normal error detecting code. One can prove \cite{Grassl:codetables} that as an error detecting/correcting code, this is the best distance achievable for a $[7,2]$ code.} and the only weight two elements of $N(S)$ are $\{Z_1Z_2,\, Z_3Z_4, \, Z_5Z_6,\, Z_5 Z_7\}$. These are all contained in the coset $\ol{Z}_1 S$ of $S$ in $N(S)$. The QET conditions \eqref{eq:qetgroups} imply the code transmutes all single qubit Pauli errors to logical phase errors $\ol{Z}_1$ on one of the logical qubits. Thus the effective distance of the code is increased to $d_{\text{eff}}=3$. Note the code with $S \cup \{\ol{Z}_1\}$ as stabiliser is an error correcting $[7,1,3]$ code.

\subsubsection*{General Case}

We now give an example of a code which encodes two qubits in six, which transmutes all single qubit Pauli errors to the set $\{I, \ol{Z}_1, \ol{Z}_2\}$ which does not form a group. The generators and logical operators are given in table \ref{tab:6qubitcode}.

\begin{table}[ht!]
\begin{center}
\begin{tabular}{L| L L L L L L }
\multirow{4}{1em}{$S$} & Y & X & Z & I & X & X \\
& Z & I & X & X & X & X \\
& Z & Z & Z & Z & I & I \\
& Z & Z & I & I & Z & Z \\\hline
\rule{0pt}{1\normalbaselineskip}
\ol{X}_1& I & X & I & X & X & I \\
\ol{X}_2& Z & I & Z & I & I & Z \\\hline
\rule{0pt}{1\normalbaselineskip}
\ol{Z}_1& Z & Z & I & I & I & I \\
\ol{Z}_2& I & I & I & I & X & X
\end{tabular}
\caption{Stabilisers and logical Pauli operators for the six qubit code.}\label{tab:6qubitcode}
\end{center}
\end{table}

All single qubit Pauli operators have non-zero syndrome, so the code is distance two as an error correcting (detecting) code. The weight two elements of $N(S)$ are $Z_1Z_2$, $Z_3Z_4$, $Z_5Z_6$, $X_5X_6$ and $Y_5Y_6$. The only sets of single qubit errors which share the same coset of $S$ are clearly those which appear in the products above. For these, we have:
\begin{equation}
\begin{aligned}
Z_1\cdot Z_2|_{\mathcal{C}} &= I \cdot \ol{Z}_1,  \\
Z_3\cdot Z_4|_{\mathcal{C}} &= I \cdot \ol{Z}_1,  \\
Z_5\cdot Z_6|_{\mathcal{C}} &= I \cdot \ol{Z}_1,
\end{aligned}
\quad\,\,
\arraycolsep=1.4pt
\begin{array}{cl}
  X_5 \cdot X_6|_{\mathcal{C}}& = I \cdot \ol{Z}_2 ,\\
   Y_5 \cdot Y_6|_{\mathcal{C}}   & =\ol{Z}_1 \cdot \ol{Z}_2,
\end{array}
\end{equation}
which explicitly demonstrates the fulfilment of the error transmuting conditions in theorem \ref{thm:qetgeneral}, and implicitly a set of the required maps $\{\pi_{\mathfrak{n}}\}$ in equation \eqref{eq:qetcgeneral1}. Thus, we have $d_{\text{eff}}=3$.

This code was found via searching over codes in standard form. The logical operators have been re-identified from the output of the standard form using the freedom from applying symplectic automorphisms. The columns (physical qubits) have been reordered. We note that, also via brute force search, a five or fewer physical qubit code with the same number of logical qubits, error transmuting property, and a minimum distance of two as an error-detecting code does not exist.

\subsection{Tileable Codes}\label{sec:tileable}

It is natural to search for QET codes which are tileable, i.e.\ possessing translational symmetry. One reason for this is that qubits are often arranged in a tiled pattern in some quantum computing architectures. Translation invariant codes may be conveniently described in terms of a unit cell, which can then be tiled as one wishes. The Laurent polynomial formalism introduced by Haah \cite{haahthesis,haah2016algebraic} provides a natural framework for this, and was previously applied to an automated search for fermionic encodings tailored to specific qubit layouts and connectivities in \cite{chien2022optimizing}.

For simplicity, we take the setup to be a 2D Euclidean lattice of sites, where there are $n$ qubits per site. A Pauli operator can be described in the following way. We define the Laurent polynomial ring $R := \bF_2[x^{\pm1},y^{\pm1}]$, and consider the module of $2n$-vectors
\begin{equation}
    R^{2n} = \left\{a= \sum_{\mathbf{k} \in \bZ^2} a_k x^{k_1} y^{k_2} \,: a_\mathbf{k} \in \bF_2^{2n}\right\},
\end{equation}
endowed with a conjugation operation
\begin{equation}
a^\dagger = \sum_{\mathbf{k} \in \bZ^2} a_k^T x^{-k_1} y^{-k_2}.
\end{equation}
The associated Pauli operator is given by:
\begin{equation}
    P(a) := \sum_{\mathbf{k} \in \bZ^2} \prod_{i=1}^n X_{i, \mathbf{k}}^{a_\bk[i]}  \prod_{i=n+1}^{2n} Z_{i, \mathbf{k}}^{a_\bk[i]},
\end{equation}
where e.g.\ $X_{i,\bk}$ denotes the Pauli $X$ operator on qubit $i$ (out of $n$) on  the $\bk^{\text{th}}$ lattice site. The translation of $T_{\bk}P(a)$ of $P(a)$ by $\bk$ in the lattice corresponds to a multiplication of $a$ by $x^k_1y^k_2$. The (anti-)commutation relations are given by:
\begin{equation}
    P(a)P(b) = (-)^{w(a,b)}P(b)P(a),
\end{equation}
where $w : R^{2n} \times R^{2n} \rightarrow R$ is the symplectic form
\begin{equation}
    w(a,b) := (a^{\dagger} \Lambda b)_{\mathbf{0}}\quad\,\, \Lambda = \begin{pmatrix} 0 & I_n \\ I_n & 0 \end{pmatrix},
\end{equation}
where the $\mathbf{0}$ subscript denotes the constant term in an element of $R$. It is easy to see the coefficient of $x^k_1y^k_2$ in $(a^{\dagger} \Lambda b)$ yields the commutation relation of $P(a)$ and a translation by $\mathbf{k}$ of $P(b)$, or a translation of $-\mathbf{k}$ of $P(a)$ with $P(b)$.

A translation invariant stabiliser is specified by $(n-k)$ elements $\{g_j\}$ of $R^{2n}$, which can be grouped in a $2n \times (n-k)$ $R$-valued matrix $\sigma$. The elements of the stabiliser are $P(a)$ for all $a$ in the $R$-submodule of $R^{2n}$ generated by $\{g_j\}$, which correspond to translations and multiplications of the Pauli operators $\{P(g_j)\}$. In order for all elements of the stabiliser to be mutually commutative, $\sigma$ must satisfy $\sigma^\dagger \Lambda \sigma  = 0$.

Further, define $\ep:=\sigma^\dagger \Lambda$. Then there is a cochain complex:
\begin{equation}
    R^{n-k}\xrightarrow[]{\sigma} R^{2n} \xrightarrow[]{\ep} R^{n-k},
\end{equation}
and $\text{ker}\, \ep$ defines those elements of $a$ such that $P(a) \in N(S)$, and the cohomology $\text{ker}\, \ep/ \text{Im}\,\sigma \cong N(S)/S$.

Let us regard the $2n$ columns $\{f_i\}$ of $\ep$ as elements of the $R$-module $R^{n-k}$. The logical operators thus correspond to elements of $R^{2n}$, with column entries $a^i=\sum_{\bk} a_{\bk}[i]x^{k_1}y^{k_2}$, such that $\sum_{i}f_i a^i = 0$. The logical operators therefore correspond to elements of the \textit{syzygy} module of the module generated by $\{f_i\}$. Note that one can in fact restrict to $R = \bF_2[x,y]$, i.e.\ regular polynomials, as it is not difficult to prove that any element of the syzygy module over $\bF_2[x,y]$ can be translated to one over $\bF_2[x^{\pm1},y^{\pm1}]$ and vice versa.

The algorithm for determining a generating set of the syzygy module may be found e.g.\ in \cite{cox2006using}. It involves computing a Gröbner basis of $\langle \{f_i\} \rangle$ using a modified version of Buchberger's algorithm. We do not include any more detail here (see e.g.\ appendix E of \cite{chien2022optimizing}), instead simply stating the results of the algorithm when required.

\subsubsection*{Example}

In the example of a QET code here, there are three physical qubits and one stabiliser per unit cell, leading to an encoding rate of $\frac{2}{3}$ (there are two logical qubits per unit cell). The stabilisers act locally in $2\times2$ subsets of unit cells, and are specified by the generator (displayed graphically in figure \ref{fig:translation_stabiliser}):
\begin{equation}\label{eq:transcodegenerator}
    g = \begin{pmatrix}
        xy \\ y+xy \\ x+xy \\ x+y \\ 1+x+xy \\ 1+y+xy
    \end{pmatrix}\quad
    \begin{matrix}
        1 \\ 2 \\ 3 \\ \hline 1 \\ 2 \\ 3
    \end{matrix}
\end{equation}

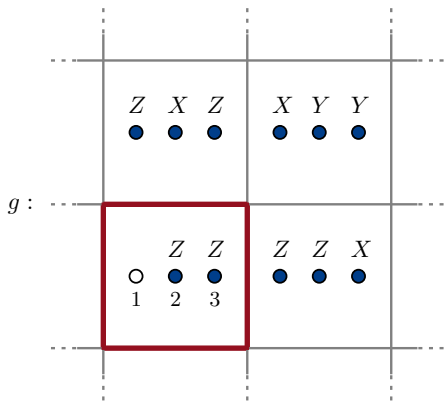
\begin{figure}[ht!]
    \centering
    \tikzset{every picture/.style={line width=0.75pt}} 
\scalebox{0.9}{
   \begin{tikzpicture}[x=0.55pt,y=0.55pt,yscale=-1,xscale=1]

\draw [color={rgb, 255:red, 128; green, 128; blue, 128 }  ,draw opacity=1 ][line width=1]    (110,90) -- (330,90) ;
\draw [color={rgb, 255:red, 148; green, 16; blue, 31 }  ,draw opacity=1 ][line width=2.25]    (110,310) -- (220,310) ;
\draw [color={rgb, 255:red, 128; green, 128; blue, 128 }  ,draw opacity=1 ][line width=1]    (330,90) -- (330,310) ;
\draw [color={rgb, 255:red, 128; green, 128; blue, 128 }  ,draw opacity=1 ][line width=1]    (220,90) -- (220,200) ;
\draw [color={rgb, 255:red, 128; green, 128; blue, 128 }  ,draw opacity=1 ][line width=1]    (110,90) -- (110,200) ;
\draw [color={rgb, 255:red, 128; green, 128; blue, 128 }  ,draw opacity=1 ][line width=1]  [dash pattern={on 1.69pt off 2.76pt}]  (110,50) -- (110,70) ;
\draw [color={rgb, 255:red, 128; green, 128; blue, 128 }  ,draw opacity=1 ][line width=1]    (110,70) -- (110,90) ;

\draw  [fill={rgb, 255:red, 0; green, 62; blue, 138 }  ,fill opacity=1 ] (190,145) .. controls (190,142.24) and (192.24,140) .. (195,140) .. controls (197.76,140) and (200,142.24) .. (200,145) .. controls (200,147.76) and (197.76,150) .. (195,150) .. controls (192.24,150) and (190,147.76) .. (190,145) -- cycle ;
\draw  [fill={rgb, 255:red, 0; green, 62; blue, 138 }  ,fill opacity=1 ] (160,145) .. controls (160,142.24) and (162.24,140) .. (165,140) .. controls (167.76,140) and (170,142.24) .. (170,145) .. controls (170,147.76) and (167.76,150) .. (165,150) .. controls (162.24,150) and (160,147.76) .. (160,145) -- cycle ;
\draw  [fill={rgb, 255:red, 0; green, 62; blue, 138 }  ,fill opacity=1 ] (130,145) .. controls (130,142.24) and (132.24,140) .. (135,140) .. controls (137.76,140) and (140,142.24) .. (140,145) .. controls (140,147.76) and (137.76,150) .. (135,150) .. controls (132.24,150) and (130,147.76) .. (130,145) -- cycle ;
\draw  [fill={rgb, 255:red, 0; green, 62; blue, 138 }  ,fill opacity=1 ] (300,145) .. controls (300,142.24) and (302.24,140) .. (305,140) .. controls (307.76,140) and (310,142.24) .. (310,145) .. controls (310,147.76) and (307.76,150) .. (305,150) .. controls (302.24,150) and (300,147.76) .. (300,145) -- cycle ;
\draw  [fill={rgb, 255:red, 0; green, 62; blue, 138 }  ,fill opacity=1 ] (270,145) .. controls (270,142.24) and (272.24,140) .. (275,140) .. controls (277.76,140) and (280,142.24) .. (280,145) .. controls (280,147.76) and (277.76,150) .. (275,150) .. controls (272.24,150) and (270,147.76) .. (270,145) -- cycle ;
\draw  [fill={rgb, 255:red, 0; green, 62; blue, 138 }  ,fill opacity=1 ] (240,145) .. controls (240,142.24) and (242.24,140) .. (245,140) .. controls (247.76,140) and (250,142.24) .. (250,145) .. controls (250,147.76) and (247.76,150) .. (245,150) .. controls (242.24,150) and (240,147.76) .. (240,145) -- cycle ;
\draw  [fill={rgb, 255:red, 0; green, 62; blue, 138 }  ,fill opacity=1 ] (300,255) .. controls (300,252.24) and (302.24,250) .. (305,250) .. controls (307.76,250) and (310,252.24) .. (310,255) .. controls (310,257.76) and (307.76,260) .. (305,260) .. controls (302.24,260) and (300,257.76) .. (300,255) -- cycle ;
\draw  [fill={rgb, 255:red, 0; green, 62; blue, 138 }  ,fill opacity=1 ] (270,255) .. controls (270,252.24) and (272.24,250) .. (275,250) .. controls (277.76,250) and (280,252.24) .. (280,255) .. controls (280,257.76) and (277.76,260) .. (275,260) .. controls (272.24,260) and (270,257.76) .. (270,255) -- cycle ;
\draw  [fill={rgb, 255:red, 0; green, 62; blue, 138 }  ,fill opacity=1 ] (240,255) .. controls (240,252.24) and (242.24,250) .. (245,250) .. controls (247.76,250) and (250,252.24) .. (250,255) .. controls (250,257.76) and (247.76,260) .. (245,260) .. controls (242.24,260) and (240,257.76) .. (240,255) -- cycle ;
\draw  [fill={rgb, 255:red, 0; green, 62; blue, 138 }  ,fill opacity=1 ] (190,255) .. controls (190,252.24) and (192.24,250) .. (195,250) .. controls (197.76,250) and (200,252.24) .. (200,255) .. controls (200,257.76) and (197.76,260) .. (195,260) .. controls (192.24,260) and (190,257.76) .. (190,255) -- cycle ;
\draw  [fill={rgb, 255:red, 0; green, 62; blue, 138 }  ,fill opacity=1 ] (160,255) .. controls (160,252.24) and (162.24,250) .. (165,250) .. controls (167.76,250) and (170,252.24) .. (170,255) .. controls (170,257.76) and (167.76,260) .. (165,260) .. controls (162.24,260) and (160,257.76) .. (160,255) -- cycle ;
\draw   (130,255) .. controls (130,252.24) and (132.24,250) .. (135,250) .. controls (137.76,250) and (140,252.24) .. (140,255) .. controls (140,257.76) and (137.76,260) .. (135,260) .. controls (132.24,260) and (130,257.76) .. (130,255) -- cycle ;
\draw [color={rgb, 255:red, 148; green, 16; blue, 31 }  ,draw opacity=1 ][line width=2.25]    (110,200) -- (220,200) ;
\draw [color={rgb, 255:red, 148; green, 16; blue, 31 }  ,draw opacity=1 ][line width=2.25]    (220,200) -- (220,310) ;
\draw [color={rgb, 255:red, 148; green, 16; blue, 31 }  ,draw opacity=1 ][line width=2.25]    (110,200) -- (110,310) ;
\draw [color={rgb, 255:red, 128; green, 128; blue, 128 }  ,draw opacity=1 ][line width=1]    (220,200) -- (330,200) ;
\draw [color={rgb, 255:red, 128; green, 128; blue, 128 }  ,draw opacity=1 ][line width=1]    (220,310) -- (330,310) ;
\draw [color={rgb, 255:red, 128; green, 128; blue, 128 }  ,draw opacity=1 ][line width=1]  [dash pattern={on 1.69pt off 2.76pt}]  (220,50) -- (220,70) ;
\draw [color={rgb, 255:red, 128; green, 128; blue, 128 }  ,draw opacity=1 ][line width=1]    (220,70) -- (220,90) ;

\draw [color={rgb, 255:red, 128; green, 128; blue, 128 }  ,draw opacity=1 ][line width=1]  [dash pattern={on 1.69pt off 2.76pt}]  (330,50) -- (330,70) ;
\draw [color={rgb, 255:red, 128; green, 128; blue, 128 }  ,draw opacity=1 ][line width=1]    (330,70) -- (330,90) ;

\draw [color={rgb, 255:red, 128; green, 128; blue, 128 }  ,draw opacity=1 ][line width=1]  [dash pattern={on 1.69pt off 2.76pt}]  (370,90) -- (350,90) ;
\draw [color={rgb, 255:red, 128; green, 128; blue, 128 }  ,draw opacity=1 ][line width=1]    (350,90) -- (330,90) ;

\draw [color={rgb, 255:red, 128; green, 128; blue, 128 }  ,draw opacity=1 ][line width=1]  [dash pattern={on 1.69pt off 2.76pt}]  (370,200) -- (350,200) ;
\draw [color={rgb, 255:red, 128; green, 128; blue, 128 }  ,draw opacity=1 ][line width=1]    (350,200) -- (330,200) ;

\draw [color={rgb, 255:red, 128; green, 128; blue, 128 }  ,draw opacity=1 ][line width=1]  [dash pattern={on 1.69pt off 2.76pt}]  (370,310) -- (350,310) ;
\draw [color={rgb, 255:red, 128; green, 128; blue, 128 }  ,draw opacity=1 ][line width=1]    (350,310) -- (330,310) ;

\draw [color={rgb, 255:red, 128; green, 128; blue, 128 }  ,draw opacity=1 ][line width=1]  [dash pattern={on 1.69pt off 2.76pt}]  (110,350) -- (110,330) ;
\draw [color={rgb, 255:red, 128; green, 128; blue, 128 }  ,draw opacity=1 ][line width=1]    (110,330) -- (110,310) ;

\draw [color={rgb, 255:red, 128; green, 128; blue, 128 }  ,draw opacity=1 ][line width=1]  [dash pattern={on 1.69pt off 2.76pt}]  (220,350) -- (220,330) ;
\draw [color={rgb, 255:red, 128; green, 128; blue, 128 }  ,draw opacity=1 ][line width=1]    (220,330) -- (220,310) ;

\draw [color={rgb, 255:red, 128; green, 128; blue, 128 }  ,draw opacity=1 ][line width=1]  [dash pattern={on 1.69pt off 2.76pt}]  (330,350) -- (330,330) ;
\draw [color={rgb, 255:red, 128; green, 128; blue, 128 }  ,draw opacity=1 ][line width=1]    (330,330) -- (330,310) ;

\draw [color={rgb, 255:red, 128; green, 128; blue, 128 }  ,draw opacity=1 ][line width=1]  [dash pattern={on 1.69pt off 2.76pt}]  (70,90) -- (90,90) ;
\draw [color={rgb, 255:red, 128; green, 128; blue, 128 }  ,draw opacity=1 ][line width=1]    (90,90) -- (110,90) ;

\draw [color={rgb, 255:red, 128; green, 128; blue, 128 }  ,draw opacity=1 ][line width=1]  [dash pattern={on 1.69pt off 2.76pt}]  (70,200) -- (90,200) ;
\draw [color={rgb, 255:red, 128; green, 128; blue, 128 }  ,draw opacity=1 ][line width=1]    (90,200) -- (110,200) ;

\draw [color={rgb, 255:red, 128; green, 128; blue, 128 }  ,draw opacity=1 ][line width=1]  [dash pattern={on 1.69pt off 2.76pt}]  (70,310) -- (90,310) ;
\draw [color={rgb, 255:red, 128; green, 128; blue, 128 }  ,draw opacity=1 ][line width=1]    (90,310) -- (110,310) ;

\draw  [color={rgb, 255:red, 148; green, 16; blue, 31 }  ,draw opacity=1 ][line width=1.5]  (110,200) -- (220,200) -- (220,310) -- (110,310) -- cycle ;

\draw (127,116) node [anchor=north west][inner sep=0.75pt]    {$Z$};
\draw (187,116) node [anchor=north west][inner sep=0.75pt]    {$Z$};
\draw (157,116) node [anchor=north west][inner sep=0.75pt]    {$X$};
\draw (237,116) node [anchor=north west][inner sep=0.75pt]    {$X$};
\draw (297,116) node [anchor=north west][inner sep=0.75pt]    {$Y$};
\draw (267,116) node [anchor=north west][inner sep=0.75pt]    {$Y$};
\draw (237,226) node [anchor=north west][inner sep=0.75pt]    {$Z$};
\draw (297,226) node [anchor=north west][inner sep=0.75pt]    {$X$};
\draw (267,226) node [anchor=north west][inner sep=0.75pt]    {$Z$};
\draw (187,226) node [anchor=north west][inner sep=0.75pt]    {$Z$};
\draw (157,226) node [anchor=north west][inner sep=0.75pt]    {$Z$};
\draw (129,265) node [anchor=north west][inner sep=0.75pt]    {\small $1$};
\draw (159,265) node [anchor=north west][inner sep=0.75pt]    {\small $2$};
\draw (189,265) node [anchor=north west][inner sep=0.75pt]    {\small $3$};

\draw (35,194) node [anchor=north west][inner sep=0.75pt]  {$g:$};

\end{tikzpicture}
}
    \caption{ The stabiliser associated to the unit cell (highlighted) for the translation invariant code. The qubit numbering within a unit cell matches that in \eqref{eq:transcodegenerator}.}
    \label{fig:translation_stabiliser}
\end{figure}

We now demonstrate that this code transmutes all single qubit Pauli errors errors to single qubit phase errors on the logical qubits, so that the code has an effective distance of $d_{\text{eff}}=3$.

Firstly, that $N(S)$ contains no single qubit errors may be verified by checking that all single qubit errors within the unit cell do not lie in $\text{ker}\,\ep$. Now, any two qubit error in $N(S)$ must necessarily be supported on a $2\times2$ subset of unit cells. This is because, assuming the contrary, either of the two single qubit errors comprising it would trigger a non-zero syndrome measurement for a $2\times2$ stabiliser containing that single qubit error.

The only two qubit logical operator is $(0,1,1,0,1,1)$, and its $x,y$ translations. It is interesting to ask if there is a generating set of the logical operators, $\text{ker} \ep$, comprising of $\{\ol{Z}_1,\ol{X}_1,\ol{Z}_2,\ol{X}_2\} \cup g$, where:
\begin{itemize}
    \item $\{\ol{Z}_1,\ol{X}_1,\ol{Z}_2,\ol{X}_2\}$ are logical Pauli operators associated to the unit cell obeying:
    \be
    \{\ol{X}_{i},\ol{Z}_{i}\}= 0
    \ee
    with all other elements commuting, and:
    \bea
    \left[T_{\bk}\ol{X}_i, \ol{X}_j\right] &=0, \quad \\
    \left[T_{\bk}\ol{Z}_i, \ol{Z}_j\right] &=0, \quad \forall\, \bk \neq \mathbf{0}.\\
    \left[T_{\bk}\ol{X}_i, \ol{Z}_j\right] &=0, \quad \\
    \eea
    This means there is a presentation of the logical Pauli group where the single qubit operators are associated to unit cells, or rather, there are two logical qubits associated locally to each unit cell.

    \item $\ol{Z}_1 = P(0,1,1,0,1,1)$. Here we may use the freedom in the identification of $N(S)/S$ with the logical Pauli group.
\end{itemize}
If such generators exist, this code transmutes all single qubit Pauli errors to the set of logical phase errors consisting of $\ol{Z}_1$ and its translations.

Algebraically the above is equivalent to finding elements $a_i, b_i$ $i=1,2$ of $R^{2n}$, which together with $g$ generate $\text{ker}\,\ep$, such that
\be
    a_i^\dagger \Lambda b_j = \delta_{ij}, \quad
    a_i^\dagger \Lambda a_j = 0, \quad
    b_i^\dagger \Lambda b_j = 0,
\ee
with $a_1 = (0,1,1,0,1,1)^{T}$. One can then set $\ol{Z}_i = P(a_i)$ and $\ol{X}_i=P(b_i)$. A particular choice is given by:
\bea
    a_2 &= (0,1,0,xy,0,1)^{T},\\
    b_1 &= (0,0,0, 1+y , 0 ,1)^{T},\\
    b_2 &= (0,1,0, x+y+xy , 1 ,0)^{T}.
\eea
The corresponding logical operators are displayed in figure \ref{fig:translation_logicals}.

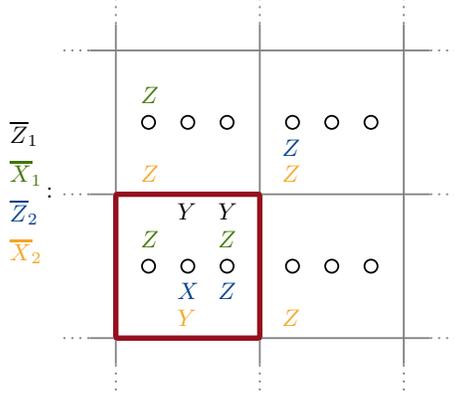
\begin{figure}[ht!]
    \centering

\tikzset{every picture/.style={line width=0.75pt}} 

\scalebox{0.9}{
\begin{tikzpicture}[x=0.55pt,y=0.55pt,yscale=-1,xscale=1]

\draw [color={rgb, 255:red, 128; green, 128; blue, 128 }  ,draw opacity=1 ][line width=0.75]    (110,90) -- (330,90) ;
\draw [color={rgb, 255:red, 148; green, 16; blue, 31 }  ,draw opacity=1 ][line width=2.25]    (110,310) -- (220,310) ;
\draw [color={rgb, 255:red, 128; green, 128; blue, 128 }  ,draw opacity=1 ][line width=0.75]    (330,90) -- (330,310) ;
\draw [color={rgb, 255:red, 128; green, 128; blue, 128 }  ,draw opacity=1 ][line width=0.75]    (220,90) -- (220,200) ;
\draw [color={rgb, 255:red, 128; green, 128; blue, 128 }  ,draw opacity=1 ][line width=0.75]    (110,90) -- (110,200) ;
\draw [color={rgb, 255:red, 128; green, 128; blue, 128 }  ,draw opacity=1 ][line width=0.75]  [dash pattern={on 0.84pt off 2.51pt}]  (110,50) -- (110,70) ;
\draw [color={rgb, 255:red, 128; green, 128; blue, 128 }  ,draw opacity=1 ][line width=0.75]    (110,70) -- (110,90) ;

\draw   (190,145) .. controls (190,142.24) and (192.24,140) .. (195,140) .. controls (197.76,140) and (200,142.24) .. (200,145) .. controls (200,147.76) and (197.76,150) .. (195,150) .. controls (192.24,150) and (190,147.76) .. (190,145) -- cycle ;
\draw   (160,145) .. controls (160,142.24) and (162.24,140) .. (165,140) .. controls (167.76,140) and (170,142.24) .. (170,145) .. controls (170,147.76) and (167.76,150) .. (165,150) .. controls (162.24,150) and (160,147.76) .. (160,145) -- cycle ;
\draw   (130,145) .. controls (130,142.24) and (132.24,140) .. (135,140) .. controls (137.76,140) and (140,142.24) .. (140,145) .. controls (140,147.76) and (137.76,150) .. (135,150) .. controls (132.24,150) and (130,147.76) .. (130,145) -- cycle ;
\draw   (300,145) .. controls (300,142.24) and (302.24,140) .. (305,140) .. controls (307.76,140) and (310,142.24) .. (310,145) .. controls (310,147.76) and (307.76,150) .. (305,150) .. controls (302.24,150) and (300,147.76) .. (300,145) -- cycle ;
\draw   (270,145) .. controls (270,142.24) and (272.24,140) .. (275,140) .. controls (277.76,140) and (280,142.24) .. (280,145) .. controls (280,147.76) and (277.76,150) .. (275,150) .. controls (272.24,150) and (270,147.76) .. (270,145) -- cycle ;
\draw   (240,145) .. controls (240,142.24) and (242.24,140) .. (245,140) .. controls (247.76,140) and (250,142.24) .. (250,145) .. controls (250,147.76) and (247.76,150) .. (245,150) .. controls (242.24,150) and (240,147.76) .. (240,145) -- cycle ;
\draw   (300,255) .. controls (300,252.24) and (302.24,250) .. (305,250) .. controls (307.76,250) and (310,252.24) .. (310,255) .. controls (310,257.76) and (307.76,260) .. (305,260) .. controls (302.24,260) and (300,257.76) .. (300,255) -- cycle ;
\draw   (270,255) .. controls (270,252.24) and (272.24,250) .. (275,250) .. controls (277.76,250) and (280,252.24) .. (280,255) .. controls (280,257.76) and (277.76,260) .. (275,260) .. controls (272.24,260) and (270,257.76) .. (270,255) -- cycle ;
\draw   (240,255) .. controls (240,252.24) and (242.24,250) .. (245,250) .. controls (247.76,250) and (250,252.24) .. (250,255) .. controls (250,257.76) and (247.76,260) .. (245,260) .. controls (242.24,260) and (240,257.76) .. (240,255) -- cycle ;
\draw   (190,255) .. controls (190,252.24) and (192.24,250) .. (195,250) .. controls (197.76,250) and (200,252.24) .. (200,255) .. controls (200,257.76) and (197.76,260) .. (195,260) .. controls (192.24,260) and (190,257.76) .. (190,255) -- cycle ;
\draw   (160,255) .. controls (160,252.24) and (162.24,250) .. (165,250) .. controls (167.76,250) and (170,252.24) .. (170,255) .. controls (170,257.76) and (167.76,260) .. (165,260) .. controls (162.24,260) and (160,257.76) .. (160,255) -- cycle ;
\draw   (130,255) .. controls (130,252.24) and (132.24,250) .. (135,250) .. controls (137.76,250) and (140,252.24) .. (140,255) .. controls (140,257.76) and (137.76,260) .. (135,260) .. controls (132.24,260) and (130,257.76) .. (130,255) -- cycle ;
\draw [color={rgb, 255:red, 148; green, 16; blue, 31 }  ,draw opacity=1 ][line width=2.25]    (110,200) -- (220,200) ;
\draw [color={rgb, 255:red, 148; green, 16; blue, 31 }  ,draw opacity=1 ][line width=2.25]    (220,200) -- (220,310) ;
\draw [color={rgb, 255:red, 148; green, 16; blue, 31 }  ,draw opacity=1 ][line width=2.25]    (110,200) -- (110,310) ;
\draw [color={rgb, 255:red, 128; green, 128; blue, 128 }  ,draw opacity=1 ][line width=0.75]    (220,200) -- (330,200) ;
\draw [color={rgb, 255:red, 128; green, 128; blue, 128 }  ,draw opacity=1 ][line width=0.75]    (220,310) -- (330,310) ;
\draw [color={rgb, 255:red, 128; green, 128; blue, 128 }  ,draw opacity=1 ][line width=0.75]  [dash pattern={on 0.84pt off 2.51pt}]  (220,50) -- (220,70) ;
\draw [color={rgb, 255:red, 128; green, 128; blue, 128 }  ,draw opacity=1 ][line width=0.75]    (220,70) -- (220,90) ;

\draw [color={rgb, 255:red, 128; green, 128; blue, 128 }  ,draw opacity=1 ][line width=0.75]  [dash pattern={on 0.84pt off 2.51pt}]  (330,50) -- (330,70) ;
\draw [color={rgb, 255:red, 128; green, 128; blue, 128 }  ,draw opacity=1 ][line width=0.75]    (330,70) -- (330,90) ;

\draw [color={rgb, 255:red, 128; green, 128; blue, 128 }  ,draw opacity=1 ][line width=0.75]  [dash pattern={on 0.84pt off 2.51pt}]  (370,90) -- (350,90) ;
\draw [color={rgb, 255:red, 128; green, 128; blue, 128 }  ,draw opacity=1 ][line width=0.75]    (350,90) -- (330,90) ;

\draw [color={rgb, 255:red, 128; green, 128; blue, 128 }  ,draw opacity=1 ][line width=0.75]  [dash pattern={on 0.84pt off 2.51pt}]  (370,200) -- (350,200) ;
\draw [color={rgb, 255:red, 128; green, 128; blue, 128 }  ,draw opacity=1 ][line width=0.75]    (350,200) -- (330,200) ;

\draw [color={rgb, 255:red, 128; green, 128; blue, 128 }  ,draw opacity=1 ][line width=0.75]  [dash pattern={on 0.84pt off 2.51pt}]  (370,310) -- (350,310) ;
\draw [color={rgb, 255:red, 128; green, 128; blue, 128 }  ,draw opacity=1 ][line width=0.75]    (350,310) -- (330,310) ;

\draw [color={rgb, 255:red, 128; green, 128; blue, 128 }  ,draw opacity=1 ][line width=0.75]  [dash pattern={on 0.84pt off 2.51pt}]  (110,350) -- (110,330) ;
\draw [color={rgb, 255:red, 128; green, 128; blue, 128 }  ,draw opacity=1 ][line width=0.75]    (110,330) -- (110,310) ;

\draw [color={rgb, 255:red, 128; green, 128; blue, 128 }  ,draw opacity=1 ][line width=0.75]  [dash pattern={on 0.84pt off 2.51pt}]  (220,350) -- (220,330) ;
\draw [color={rgb, 255:red, 128; green, 128; blue, 128 }  ,draw opacity=1 ][line width=0.75]    (220,330) -- (220,310) ;

\draw [color={rgb, 255:red, 128; green, 128; blue, 128 }  ,draw opacity=1 ][line width=0.75]  [dash pattern={on 0.84pt off 2.51pt}]  (330,350) -- (330,330) ;
\draw [color={rgb, 255:red, 128; green, 128; blue, 128 }  ,draw opacity=1 ][line width=0.75]    (330,330) -- (330,310) ;

\draw [color={rgb, 255:red, 128; green, 128; blue, 128 }  ,draw opacity=1 ][line width=0.75]  [dash pattern={on 0.84pt off 2.51pt}]  (70,90) -- (90,90) ;
\draw [color={rgb, 255:red, 128; green, 128; blue, 128 }  ,draw opacity=1 ][line width=0.75]    (90,90) -- (110,90) ;

\draw [color={rgb, 255:red, 128; green, 128; blue, 128 }  ,draw opacity=1 ][line width=0.75]  [dash pattern={on 0.84pt off 2.51pt}]  (70,200) -- (90,200) ;
\draw [color={rgb, 255:red, 128; green, 128; blue, 128 }  ,draw opacity=1 ][line width=0.75]    (90,200) -- (110,200) ;

\draw [color={rgb, 255:red, 128; green, 128; blue, 128 }  ,draw opacity=1 ][line width=0.75]  [dash pattern={on 0.84pt off 2.51pt}]  (70,310) -- (90,310) ;
\draw [color={rgb, 255:red, 128; green, 128; blue, 128 }  ,draw opacity=1 ][line width=0.75]    (90,310) -- (110,310) ;

\draw  [color={rgb, 255:red, 148; green, 16; blue, 31 }  ,draw opacity=1 ][line width=1.5]  (110,200) -- (220,200) -- (220,310) -- (110,310) -- cycle ;

\draw (127,116) node [anchor=north west][inner sep=0.75pt]  [color={rgb, 255:red, 65; green, 117; blue, 5 }  ,opacity=1 ]  {$Z$};
\draw (186,266) node [anchor=north west][inner sep=0.75pt]  [color={rgb, 255:red, 0; green, 62; blue, 138 }  ,opacity=1 ]  {$Z$};
\draw (155,266) node [anchor=north west][inner sep=0.75pt]  [color={rgb, 255:red, 0; green, 62; blue, 138 }  ,opacity=1 ]  {$X$};
\draw (235,156) node [anchor=north west][inner sep=0.75pt]  [color={rgb, 255:red, 0; green, 62; blue, 138 }  ,opacity=1 ]  {$Z$};
\draw (235,286) node [anchor=north west][inner sep=0.75pt]  [color={rgb, 255:red, 245; green, 166; blue, 35 }  ,opacity=1 ]  {$Z$};
\draw (235,176) node [anchor=north west][inner sep=0.75pt]  [color={rgb, 255:red, 245; green, 166; blue, 35 }  ,opacity=1 ]  {$Z$};
\draw (186,205) node [anchor=north west][inner sep=0.75pt]    {$Y$};
\draw (155,205) node [anchor=north west][inner sep=0.75pt]    {$Y$};
\draw (127,226) node [anchor=north west][inner sep=0.75pt]  [color={rgb, 255:red, 65; green, 117; blue, 5 }  ,opacity=1 ]  {$Z$};
\draw (186,226) node [anchor=north west][inner sep=0.75pt]  [color={rgb, 255:red, 65; green, 117; blue, 5 }  ,opacity=1 ]  {$Z$};
\draw (155,286) node [anchor=north west][inner sep=0.75pt]  [color={rgb, 255:red, 245; green, 166; blue, 35 }  ,opacity=1 ]  {$Y$};
\draw (127,176) node [anchor=north west][inner sep=0.75pt]  [color={rgb, 255:red, 245; green, 166; blue, 35 }  ,opacity=1 ]  {$Z$};

\draw (27,142.4) node [anchor=north west][inner sep=0.75pt]    {$\overline{Z}_{1}$};
\draw (27,172.4) node [anchor=north west][inner sep=0.75pt]  [color={rgb, 255:red, 65; green, 117; blue, 5 }  ,opacity=1 ]  {$\overline{X}_{1}$};
\draw (27,202.4) node [anchor=north west][inner sep=0.75pt]  [color={rgb, 255:red, 0; green, 62; blue, 138 }  ,opacity=1 ]  {$\overline{Z}_{2}$};
\draw (27,232.4) node [anchor=north west][inner sep=0.75pt]  [color={rgb, 255:red, 245; green, 166; blue, 35 }  ,opacity=1 ]  {$\overline{X}_{2}$};

\draw (52,193) node [anchor=north west][inner sep=0.75pt]  {$\,:$};

\end{tikzpicture}
}
    \caption{The logical operators associated to the unit cell for the translation invariant code.}
    \label{fig:translation_logicals}
\end{figure}

To check that $\{a_i, b_i, g\}$ generate $\text{ker}\,\ep$, we first compute a generating set of the syzygy module of the module generated by the columns of $\ep$. Without loss of generality, we may compute the syzygy module of $\ep$ multiplied by some monomial of non-negative degree. Note that $xy \ep$ contains $1$ as a column, and thus a Gröbner basis of $xy \ep$ is just $\{1\}$. The algorithm for computing syzygies then tells us that the columns of $I - B xy \ep$, where $B$ is any element $R^{2n}$ such that $\ep \cdot B = 1$ (which may be found by a polynomial division algorithm), generate the syzygy module. It is then not difficult to check that these syzygy generators are contained in the $R$-module generated by $\{a_i,b_i, g\}$.

In summary, the translation invariant code with stabiliser and logical operators in figures \ref{fig:translation_stabiliser} and \ref{fig:translation_logicals}, transmutes all single qubit Pauli errors to single qubit phase errors on half the logical qubits.

By an exhaustive search, it can be checked that there are no error correcting codes (in the usual sense) of distance $3$ with three qubits and one stabiliser per unit cell, such that the stabiliser generators are supported in a $2 \times 2$ grid. Note that, similarly to the example in \ref{sec:lowqubitexample}, if one includes $\ol{Z}_1$ and its translations in the stabiliser group then one obtains an error correcting code.

Note that it is possible to find a single-error correcting code with $n=2$ and a single stabiliser supported on a $2\times2$ grid per unit cell. For example:
\begin{equation}
    g = (xy, xy+y, 1+x+y, 1+xy)^{T}.
\end{equation}
This has a lower encoding rate $\frac{1}{2}$. In fact, translation invariant codes  with encoding rate $\frac{1}{2}$, which are also fermionic encodings, have been found with distance $d=4,5,6,7$ \cite{chen2022error}. For these, the stabiliser generators are higher weight and supported on a larger neighbourhood of the unit cell.

\subsection{17 Qubit CSS Code }\label{sec:CSScodes}

We now give an example of a CSS code \cite{Calderbank:1995dw, Steane:1995vv} encoding two logical qubits in seventeen physical qubits, which transmutes all weight two Pauli errors to single qubit phase errors.

A CSS code is specified by two classical error correcting codes $C_i$,  $i=1,2$, with parameters $[n, k_i, d_i]$. The dual of each code must be contained in the other code: $C_1^{\perp} \subset C_2$, and $C_2^{\perp}\subset C_1$. A set of generators of the stabiliser of the quantum code is specified by the rows of the parity check matrix $P_1$ of $C_1$, by replacing each $1$ with a $Z$, together with the rows of $P_2$, replacing each $1$ by an $X$. The parameters of the quantum code are $[n, k_1+k_2-n, \min(d_1,d_2)]$. If $d_1 \neq d_2$, the code may also be regarded as an $[n,k_1+k_2-n, d_1 / d_2]$ \textit{asymmetric} code  \cite{Ioffe_2007}, which corrects $\lfloor \frac{d_1-1}{2}\rfloor$ $X$ errors, $\lfloor \frac{d_2-1}{2}\rfloor$ $Z$ errors, and $\lfloor \frac{\min(d_1,d_2)-1}{2}\rfloor$ $Y$ errors. This is particularly useful in the case of a biased noise model $\cE$ where, for instance, the probability of a phase error is greater than that of a bit flip.

In our example, we take $C_1$ to be the $[17,9,5]$ binary QR code \cite{prange1957cyclic}, which is a cyclic BCH code with generator polynomial:
\begin{equation}
    g(x) = 1+x^3+x^4+x^5+x^8.
\end{equation}
A codeword corresponds to the coefficients of a polynomial, and a generating set for the codewords is given by $g(x), xg(x),\ldots, x^8g(x)$. The generator matrix for the code has rows given by these codewords. The parity check matrix may be obtained in the standard way. See \cite{aly2007quantum} for more on BCH codes.

Instead of describing $C_2$ directly, we instead describe $C_2^{\perp}$, which must be a subcode of $C_1$. The even subcode of $C_1$ may be generated by $\tilde{g}(x),x\tilde{g}(x),\ldots, x^7\tilde{g}(x)$ where $\tilde{g} := (1+x^5)g(x)$. We take $C_2^{\perp}$ to be generated by the first seven of these generators, i.e.\  $\tilde{g}(x),x\tilde{g}(x),\ldots, x^6\tilde{g}(x)$. These yield the parity check matrix of $C_2$.

One can check that $C_2$ has distance 3, thus the code $\text{CSS}(C_1,C_2)$ as an asymmetric error correcting code has parameters $[17,2,3/5]$. However, one can check that all 3 qubit $X$ errors lie in the coset of $N(S)/S$ corresponding to $1+x^3+x^7$
and all 4 qubit $X$ errors in that of
$x^3+x^6+x^{10}+x^{11}$. These cosets are distinct and thus correspond to different logical operators which, as pure $X$ operators, commute. We can identify them with the logical operators $\{\ol{X}_1, \ol{X}_2\}$, or equivalently $\{\ol{Z}_1, \ol{Z}_2\}$.

Thus: if the admissible error set is defined to be these two single logical qubit $\ol{X}$ (or $\ol{Z}$) operators, the code can transmute all Pauli errors up to weight two on the physical qubits. The distance of the code is effectively increased to $d_{\text{eff}}=5$. We expect there are similar constructions for other asymmetric CSS codes.

The above code was found by searching over codes $C_2$, such that $C_2^\perp$ is contained in the classical $[17,9,5]$ QR code. The search also demonstrated that it is not possible for $C_2$ to be cyclic for the same $n, k_2$. There are codes with better parameters as asymmetric CSS codes, however it is not necessarily true they will have error transmuting properties (e.g.\ the $[15,3,3/5]$ code \cite{aly2008asymmetric}  does not).

\subsection{Operations on QET Codes}

We briefly discuss two constructions, which produce new error correcting codes from old ones, that can also be applied to error transmuting codes.

\subsubsection*{Concatenation}

Suppose we have an error transmuting code $\cC_1$ of parameters $[n_1,k]$, specified by a stabiliser $S$, for an admissible error set $\ol{M}$. We know from section \ref{sec:effcodedistance} that the effective distance $d_{\text{eff}}$ is bounded below by the minimum weight of an element in $N(S) \backslash M$, which we denote $\tilde{d}_{\text{eff}}$. Suppose we concatenate the code with a code $\cC_2$ with parameters $[n_2,1]$ which has error \textit{correcting} distance $d_2$, by encoding each of the $n_1$ physical qubits of $\cC_2$ as the logical qubit of a copy of $\cC_2$. The resulting code $\cC'$ has parameters $[n_1n_2,k_1]$ with an effective distance at least $\tilde{d}_{\text{eff}}d_2$.

To see this note that the logical operators of the resulting code $\cC'$, whose stabiliser we denote $S'$, can be generated by the same logical operators as $\cC_1$ but with Pauli operators in the tensor products replaced by the corresponding encoded logical operators of $\cC_2$. Using this, the isomorphism $\cA: \ol{\cG}_{k} \rightarrow N(S)/S$ can be lifted naturally to a corresponding isomorphism $\cA': \ol{\cG}_{k} \rightarrow N(S')/S'$, which defines the lift $M'$ of the admissible error set $\ol{M}$ now in $N(S')$. The elements of $N(S')\backslash M$ must have minimum weight $\tilde{d}_{\text{eff}}d_2$ since they must have weight $d_2$ on at least $\tilde{d}_{\text{eff}}$ blocks of $n_2$ qubits.

\subsubsection*{Majorana Fermion Code Construction}

The authors of \cite{Bravyi:2010de} demonstrated how a $[n,k,d]$ qubit stabiliser code $\cC$ could be mapped to a $[4n,2k,2d]$ weakly self-dual CSS code $\text{CSS}(\cC)$. `Weakly self-dual' refers to the fact that the classical codes $\cC_1$, $\cC_2$ specifying the $X$ and $Z$ parts of $\text{CSS}(\cC)$ are identical. Here $d$ is the usual error correcting distance of the code.

Without going into the details of the construction, there is a one-to-one correspondence between each element of $N(S)$ for $\cC$ and a set of $2^n$ pure $X$, and a set of $2^n$ pure $Z$ logical operators of $\text{CSS}(\cC)$. Stabilisers are mapped to/from stabilisers. If the weight of the element of $N(S)$ is $w$, the corresponding logical operators of $\text{CSS}(\cC)$ are weight $2w$. Conversely, a general logical operator $F$ of $\text{CSS}(\cC)$ can be decomposed as $F= c F_X \cdot F_Y$ where $c \in \{\pm 1,\pm i\}$. By the construction, both $F_X$ and $F_Y$ each correspond to a logical operator of $\cC$, of half the weight.

Suppose $\cC$ is an error transmuting code for an admissible error set $\ol{M}$. As before its effective distance is bounded below by $\tilde{d}_{\text{eff}}$, the minimum weight of an element of $N(S) \backslash M$. Let $M_X$ and $M_Z$ be the set of pure $X$ and pure $Z$ logical operators mapped to under the correspondence. Define $M' = \{\pm1, \pm i\} M_X M_Z$. Then one has that:
\be
\min_{E \in M' } \text{weight}(E) = 2\tilde{d}_{\text{eff}},
\ee
so that $\text{CSS}(\cC)$ is an error transmuting code for the logical error set $M'$ of effective distance at least~$2\tilde{d}_{\text{eff}}$.
\enlargethispage{1em}

\paragraph{Further directions} There are many questions which follow naturally from the material presented in this work. Of course, it would be desirable to develop constructions of \textit{families} of codes, whose error transmuting capabilities scale nicely as the number of physical and logical qubits are increased. Perhaps recent developments in the theory of LDPC codes and product constructions \cite{breuckmann2021quantum} can be applied to produce desirable QET codes.  With an eye towards applications, the most pertinent questions are of course whether or not there exist other algorithms which are resilient to a particular form of residual noise on the logical qubits, or other encodings of physical systems for which there are admissible errors corresponding to natural, physical noise (and are able to transmute higher weight errors).

\paragraph{Acknowledgements} The authors would like to thank Charles Derby, Joel Klassen and Ashley Montanaro for many helpful discussions in the process of completing this work.

\bibliographystyle{ieeetr}
\bibliography{qet}

\end{document}

%% file: preamble.tex
\usepackage{amsfonts}
\usepackage{amscd}
\usepackage{amssymb}
\usepackage{amsmath,bbm}
\usepackage{graphicx}
\usepackage{epsfig}
\usepackage{latexsym}
\usepackage{mathtools}
\usepackage{hyperref}
\usepackage{tikz-cd}
\usepackage[vcentermath]{youngtab}

\definecolor{cardinal}{rgb}{0.6,0,0}
\definecolor{darkgreen}{rgb}{0,0.5,0}
\definecolor{golden}{rgb}{0.92, 0.7, 0}
\definecolor{midnight}{rgb}{0, 0, 0.5}
\definecolor{darkblue}{rgb}{0.2, 0, 0.8}

\hypersetup{colorlinks=true, linkcolor=black,citecolor=darkblue,filecolor=black,urlcolor=black,}
    
\def \be  {\begin{equation}}
\def \ee  {\end{equation}}
\def \bea {\begin{equation}\begin{aligned}}
\def \eea {\end{aligned}\end{equation}}
\def \ba  {\begin{eqnarray}}
\def \ea  {\end{eqnarray}}
\def \bb  {\begin{thebibliography}}
\def \eb  {\end{thebibliography}}
\def \lab #1 {\label{#1}}

\newcommand\ep{\epsilon}

\newcommand\ol{\overline}
\newcommand\cA{\mathcal{A}}
\newcommand\cB{\mathcal{B}}
\newcommand\cC{\mathcal{C}}

\newcommand\cE{\mathcal{E}}

\newcommand\cG{\mathcal{G}}
\newcommand\cH{\mathcal{H}}

\newcommand\cM{\mathcal{M}}

\newcommand\cR{\mathcal{R}}

\newcommand\al{\alpha}

\newcommand\fl{\mathfrak{l}}

\newcommand\bC{\mathbb{C }}
\newcommand\bF{\mathbb{F }}

\newcommand\bZ{\mathbb{Z}}

\newcommand\fm{\mathfrak{m}}

\newcommand\fn{\mathfrak{n}}

\newcommand\bk{\mathbf{k}}